\documentclass[9pt,shortpaper,twoside,web]{ieeecolor}
\usepackage{generic}
\usepackage{cite}
\usepackage{amsmath,amssymb,amsfonts}
\usepackage{mathrsfs}
\usepackage{algorithmic}
\usepackage{graphicx}
\usepackage{textcomp}
\usepackage{color}
\usepackage{latexsym}
\usepackage{amsfonts,amssymb}
\usepackage{textcomp}
\usepackage{bbm}
\usepackage{dsfont}
\usepackage{fancyhdr}
\usepackage{caption}

\renewenvironment{proof}{\vspace{.1cm}\noindent{\sc Proof.}\hspace{0.10cm}\,\,}{$\hfill\Box$\vspace{.1cm}} 
\newtheorem{theorem}            {Theorem}[section] 

\newtheorem{definition}         [theorem]{Definition}

\newtheorem{lemma}              [theorem]{Lemma} 
\newtheorem{proposition}		[theorem]{Proposition} 
\newtheorem{corollary}		[theorem]{Corollary}
 
\newtheorem{example}		[theorem]{Example}
\newtheorem{remark}	      [theorem]{Remark}

\newcommand{\hygreen}[1]{{\color{magenta}{#1}}}
\newcommand{\dd}   {{\rm d}\hbox{\hskip 0.5pt}}

\newcommand{\bbm}[1]{\left[\begin{matrix} #1 \end{matrix}\right]}
\newcommand{\sbm}[1]{\left[\begin{smallmatrix} #1
   \end{smallmatrix}\right]}

\def\BibTeX{{\rm B\kern-.05em{\sc i\kern-.025em b}\kern-.08em
    T\kern-.1667em\lower.7ex\hbox{E}\kern-.125emX}}
\begin{document}
\title{Contraction analysis of time-varying DAE systems via auxiliary ODE systems}
\author{Hao Yin, Bayu Jayawardhana and Stephan Trenn
\thanks{The authors are with the Jan C. Willems Center for Systems and Control, Faculty of Science and Engineering, University of Groningen, 9747AG Groningen, The Netherlands.  {\tt\small hao.yin@rug.nl, b.jayawardhana@rug.nl and s.trenn@rug.nl}.}
\thanks{This work was supported by China Scholarship Council.}
}
\maketitle
\begin{abstract}
This paper studies the contraction property of time-varying differential-algebraic equation (DAE) systems by embedding them to higher-dimension ordinary differential equation (ODE) systems. The first result pertains to the equivalence of the contraction of a DAE system and the uniform global exponential stability (UGES) of its variational DAE system. Such equivalence inherits the well-known property of contracting ODE systems on a specific manifold. Subsequently, we construct an auxiliary ODE system from a DAE system whose trajectories 
encapsulate those of the corresponding variational DAE system. Using the auxiliary ODE system, a sufficient condition for contraction of the time-varying DAE system is established by using matrix measure 
which allows us to estimate an lower bound on the parameters of the auxiliary system. Finally, we apply the results to analyze the stability of time-invariant DAE systems, and to design observers for time-varying ODE systems. 
\end{abstract}

\section{Introduction}
As a generalization of 
ordinary differential equation (ODE) systems, differential-algebraic equation (DAE) systems have been studied for the past decades due to their relevance in representing numerous modern engineering systems with constraints. Some well-known examples of such engineering systems are electrical networks with Kirchhoff's laws \cite{Monshizadeh2018} and mechanical systems with rigid body constraints \cite{Rabier2000}. 
The DAE systems can also be used to model power systems \cite{Hill1990} and chemical processes \cite{Kumar1999}. Typically, DAE systems consist of a set of differential equations describing the system dynamics, and a set of algebraic equations describing the constraints. Solving DAE systems is more challenging than solving ODE systems due to the implicit relationship between 
the differential equations and algebraic equations. DAE systems can be classified as either index-1 \cite{Hill1990} or higher index \cite{Gerdts2003} by the degree of the highest derivative of the algebraic equation. Index-1 DAE systems are particularly important in control theory, because the algebraic restriction can be substituted directly into the dynamics. Consequently, they can numerically be solved by using standard techniques for solving ODE \cite{Takens1976}, which allows numerical simulations to be carried out 
straightforwardly 
\cite{Shampine1999}.

The analysis of DAE systems has been widely studied in the literature \cite{Groß2016,Hill1990,Nugroho2022,Nugroho2023,Praprost1994,Milano2012,Persis2016} with a large body of works concerned with index-1 DAE systems. In \cite{Groß2016}, the authors show that any solvable DAE power systems are 
of index one. By this result, robust $H_{\infty}$ observers for power networks are presented in \cite{Nugroho2022}, and a load- and renewable-following control approach for power system is proposed in \cite{Nugroho2023}. There are several methods for analyzing the stability of DAE systems, including the Lyapunov method \cite{Hill1990,Persis2016}, the energy-based methods \cite{Praprost1994}, as well as, the eigenvalue analysis \cite{Milano2012}. In \cite{Hill1990}, the DAE systems are embedded in reduced ODE systems, where Lyapunov method can be applied to get the Lyapunov stability. 
In \cite{Persis2016}, an incremental Lyapunov function is applied to analyze the asymptotic stability and optimal resource allocation for a network preserved microgrid model with active and reactive power loads. In \cite{Praprost1994}, 
bifurcation theory is used to characterize the stability boundary for power systems in DAE, and an energy function method is developed to guarantee both rotor angular stability and voltage stability. A method for defining the small-signal stability of delay DAE systems based on eigenvalue analysis and an approximation of the characteristic equation at equilibrium points are proposed in \cite{Milano2012}. 

In all of the above results, the properties of time-invariant DAE systems are analyzed, 
while the extension of these results to the time-varying DAE systems remains non-trivial. The exponential stability and  robustness of linear time-varying DAE systems with index-1 are studied in \cite{Berger2014,Chyan2008}. In \cite{Chyan2008}, the Bohl
exponent theory 
for stability analysis of ODEs is extended to DAEs. 
When there are perturbations in systems' matrices, a stability radius has been investigated in \cite{Berger2014}, which includes the lower bound computation of the stability radius. 

As one of the stability analysis methods for time-varying systems, which has gained popularity in recent years, 
contraction analysis focuses on the relative trajectory of the nonlinear time-varying system rather than a specific equilibrium point. There are many methods to analyze the contractivity of nonlinear time-varying ODE systems in literature, such as, \cite{Lohmiller1998,Forni2013,Barabanov2019,Yin2023,Andrieu2021} among many others. An ODE system is contracting if and only if the associated variational system is uniformly globally exponentially stable (UGES) \cite{Barabanov2019}. In \cite{Lohmiller1998}, the contraction property can be guaranteed if the largest eigenvalue of the symmetric part of the variational systems is uniformly strictly negative. Finsler–Lyapunov functions are introduced in \cite{Forni2013} to analyze the incremental exponential stability of the system. Mode-dependent Lyapunov functions for contracting switched nonlinear systems are presented in \cite{Yin2023}. In \cite{Andrieu2021},  the focus is on investigating transverse exponential stability (a generalized notion of contraction) by employing a Lyapunov matrix transversal equation. In the context of DAE systems, the contraction analysis thereof has recently been presented in \cite{Nguyen2021}. In \cite{Nguyen2021}, the authors proved that if the algebraic equation satisfies some sufficient conditions, the contractivity of time-invariant DAE systems can be obtained through exponential stability analysis of the corresponding reduced variational ODE systems using matrix measure. In Section \ref{sectionode}, we show that this approach can be restrictive and is not applicable to 
some time-varying DAE systems.

In this paper, we analyze the contraction property of index-1 nonlinear time-varying DAE systems by using an ODE approach. As our first main result, we establish that the uniform global exponential stability (UGES) of the  variational DAE dynamics is a sufficient and necessary condition to the contractivity of the original DAE systems. This condition can be viewed as a DAE counterpart of the results for ODE presented in 
\cite[Prop.~1]{Barabanov2019}. Subsequently, we construct an auxiliary ODE system whose convergence property can encapsulate the same 
property of the variational DAE dynamics. With this construction, we can analyze the contractivity of the DAE systems by applying conventional control theory (such as Lyapunov approach, matrix measure method) to the auxiliary ODE system. As our second contribution, we provide sufficient conditions on the contraction of nonlinear time-varying DAE system by using matrix measure method to analysis the UGES property of the auxiliary ODE system. In general, these conditions ensure the contraction of the DAE system without analysing its reduced system. 
As our third contribution, we employ these conditions to design observers for time-varying ODE systems by treating the output as an algebraic equation. Furthermore, we investigate the exponential stability of time-invariant DAE systems by ensuring that the DAE system is contractive and the equilibrium lies within its trajectory set.

The paper is organized as follows. In Section 2, we present preliminaries and problem formulation. Our main results are proposed in Section 3, where we present 
necessary and sufficient conditions for the contractivity of time-varying DAE systems, and the ODE approach. The numerical simulations and applications are provided in Section 4 and the conclusions are given in Section 5. 

\section{Preliminaries and problem formulation}

Throughout this paper, we consider the following nonlinear time-varying DAE systems 
\begin{equation}\label{daet}
 \left\{\begin{matrix}
\dot{w}=f(t,w,z),\\ 
0=g(t,w,z),
\end{matrix}\right.   
\end{equation}
where $w(t)\in \mathbb{R}^{n}$ is the state vector, $z(t)\in \mathbb{R}^{m}$
refers to the algebraic vector, $f:\mathbb{R}_+\times \mathbb{R}^n \times \mathbb{R}^m\rightarrow \mathbb{R}^n$ is the vector field, and $g:\mathbb{R}_+\times \mathbb{R}^n \times \mathbb{R}^m\rightarrow \mathbb{R}^m$ describes the algebraic manifold. We assume that $f$ is continuously differentiable and $g$ is twice-continuously differentiable. In this note, we consider only the continuously differentiable solutions $
w(t)=\varphi(t_0,w_0,z_0)$ and $z(t)=\psi (t_0,w_0,z_0)$ of \eqref{daet} with admissible
initial conditions $(w_0,z_0)\in \mathbb{R}^n\times \mathbb{R}^m$ satisfying the algebraic constraint
\begin{equation}\label{ini1}
 g(t,w_0,z_0)=0.  
\end{equation}
We assume 
the DAE system \eqref{daet} is of index-1, i.e.\ 
the (partial) Jacobian matrix 
$\frac{\partial g}{\partial z}
(t,w,z)\in\mathbb{R}^{m\times m}
$ is invertible 
for all $(t,w,z)$
. This assumption guarantees the existence
and uniqueness of 
a local 
solution to \eqref{daet} for any initial condition satisfying \eqref{ini1} (see also, \cite{Hill1990}). 
Throughout the paper, we will assume that every local solution can be extended to a solution defined on the whole time domain $[0,\infty)$.

Note that the index-1 assumption allows to apply the implicit function theorem to solve the constraint \eqref{ini1} for $z_0$ for any given $(t_0,w_0)$ and in the following we will therefore write $z_0(w_0)$ to denote the unique value for $z_0$ which satisfies \eqref{ini1} for an arbitrarily given $w_0$ (we omit the dependency on $t_0$ as we consider the initial time as fixed in the following analysis).

\begin{definition}\label{d1}
A time-varying DAE system 
\eqref{daet} is called {\em contracting} if there exists positive numbers $c$ and $\alpha$ such that for any pair of solutions $W_{i}(t)=\sbm{w_i(t) \\z_i(t)}\in \mathbb{R}^n$ of \eqref{daet} with 
$i=1,2$, we have
\begin{equation} \label{eq:d1}
\begin{aligned}
\|
W_{1}(t)-W_{2}(t)
\|\leq ce^{-\alpha t}\|
W_{1}(t_{0})-W_{2}(t_{0})
\|, \quad \forall t\geq t_0.
\end{aligned}
\end{equation}
\end{definition}
\vspace{0.3cm}

In order to study contractivity of the DAE system \eqref{daet}, we will analyse the (uniform) stability of the corresponding variational DAE systems. The variational system of system \eqref{daet} is given by 
\begin{equation}\label{vari}
 \left\{\begin{matrix}
\dot{\xi }=\frac{\partial f}{\partial w}(t,w(t),z(t))\cdot \xi +\frac{\partial f}{\partial z}(t,w(t),z(t))\cdot\nu,\\ 
0=\tfrac{\partial g}{\partial w}(t,w(t),z(t))\cdot\xi +\frac{\partial g}{\partial z}(t,w(t),z(t))\cdot\nu,
\end{matrix}\right.   
\end{equation}
where $w(t)$ and $z(t)$ are solutions of \eqref{daet}. We omit
the explicit parametrization $(t,w,z)$ whenever it is clear
from the context. 
\begin{definition}\label{d2}
The variational DAE system \eqref{vari} is called uniformly globally exponentially stable (UGES), if there exist positive numbers $c$, $\alpha$ (independent of the solution $W(\cdot)$, but may be dependent on $t_0$) such that for every solution $\Xi(t):=\sbm{\xi(t)\\\nu(t)}\in\mathbb{R}^n$ of \eqref{vari} the inequality
\begin{equation} \label{eq:d3}
\begin{aligned}
\|\Xi(t)\| \leq ce^{-\alpha t}\|\Xi(t_{0})\|, 
\end{aligned}
\end{equation}
holds for all $t\geq t_0$.
\end{definition}

Note that in the above definition and in the remainder of the paper, the norm $\|\cdot\|=\|\cdot\|_p$ is some $p$-norm, $p\geq 1$, where we usually omit the index $p$, unless a specific case is discussed; while the exponential rate $\alpha$ in \eqref{eq:d3} does not depend on the chosen norm, the constant $c$ in general depends on it.

By invertibility of the Jacobian matrix $\frac{\partial g}{\partial z}$, we can express the reduced system of equation (6) in the following form:
\begin{equation} \label{eq:re}
\dot{\xi}=\left(\frac{\partial f}{\partial w}-\frac{\partial f}{\partial z}\left[\frac{\partial g}{\partial z}\right]^{-1}\tfrac{\partial g}{\partial w}\right)\xi.
\end{equation}
In \cite[Proposition 1]{Nguyen2021}, a sufficient condition is introduced to analyse the contractivity of time-invarient DAE systems by analysing its reduced systems \eqref{eq:re}. However, this approach may not be applicable in the time-varying case. As a simple example, consider the following time-varying DAE system
\begin{equation}\label{exam1}
\left\{\begin{matrix}
\dot{w}=-w+e^{-3t}z,\\ 
0=e^{3t}w+z,
\end{matrix}\right.   
\end{equation}
which has the same form as its variational DAE system. 
Its reduced system is given by $\dot{w}=-2w$, which is a contractive system. Consequently, the trajectories of the system are given by $w(t)=w_0e^{-2t}$ and $z(t)=-w_0e^{t}$, which shows the DAE system is not contracting ($z$ is not contracting). This is due to the unboundedness of $\left[\frac{\partial g}{\partial z}\right]^{-1}\tfrac{\partial g}{\partial w}=e^{3t}$, which implies that the method proposed in \cite{Nguyen2021} is not applicable when dealing with time-varying systems. Notice that, in this particular case, the contraction analysis problem can be effectively addressed by comparing the exponential rate of the reduced system with an exponential bound on $\left[\frac{\partial g}{\partial z}\right]^{-1}\tfrac{\partial g}{\partial w}$. However, it is worth acknowledging that obtaining information about the reduced system or establishing a bound for $\left[\frac{\partial g}{\partial z}\right]^{-1}\tfrac{\partial g}{\partial w}$ can be a challenging task for certain systems.

The objective of this paper is to provide a sufficient condition that guarantees the contractivity of time-varying DAE systems, even in situations where prior knowledge about the DAE system is unavailable.


Before stating our main results, let us recall the definition of matrix measure \cite{Desoer1972}, which
is also known as logarithmic norm as follows. For a given matrix $A\in \mathbb{R}^{n\times n}$ 
the matrix measure $\mu_{p}(A)$ w.r.t.\ to the (induced) $p$-norm $\|\cdot\|_p$ 
is defined by
\begin{equation}\label{mm}
 \mu_{p}(A):=\underset{h\rightarrow 0^+}{\lim}\frac{||I+hA||_{p}-1}{h}.
\end{equation}
It follows that when $p=1$, $2$ or $\infty$, we have
 \begin{equation}\label{mm1}
 \begin{aligned}
&\mu_1(A)=\max_j\Big(a_{jj}+\underset{i\neq j}{\sum }|a_{ij}|\Big),
\\& \mu_2(A)=\max_i\Big(\lambda _{i}\left\{\tfrac{A+A^{\top}}{2}\right\}\Big), \text{ or }
\\& \mu_{\infty}(A)=\max_i\Big(a_{ii}+\underset{j\neq i}{\sum }|a_{ij}|\Big),
\end{aligned}
\end{equation}
respectively. In this paper, all the norms are defined using a $p$-norm.

\section{Main results}
In this section, we firstly establish an equivalent relationship 
between the contraction of a DAE system and the uniform global exponential stability (UGES) of its variational DAE system. Secondly, we construct an auxiliary ODE system that encapsulates the behaviors of the variational DAE system. Thirdly, a sufficient condition is presented that guarantees the UGES of the auxiliary ODE system and is numerically implementable. 

\subsection{A necessary and sufficient condition} 
\begin{proposition}\label{proposition1}
The DAE system \eqref{daet} is contracting if and only if the corresponding variational DAE system \eqref{vari} is UGES.
\end{proposition}

\begin{proof}
Let us establish a relationship between the solutions of \eqref{daet} and those of \eqref{vari}. Let 
$\sbm{w(t) \\ z(t)}=\sbm{\varphi(t,w_0,z_0(w_0)) \\ \psi (t,w_0,z_0(w_0))}$ and $\sbm{\hat{w}(t)\\ \hat{z}(t)}=\sbm{\varphi(t,w_0+\delta \xi _0,z_0(w_0+\delta \xi_0)) \\ \psi (t,w_0+\delta \xi_0,z_0(w_0+\delta \xi_0))}$ be two trajectories of \eqref{daet} with initial conditions $\sbm{w_0 \\ z_0(w_0)}$
and 
$\sbm{w_0+\delta \xi_0\\ z_0(w_0+\delta \xi_0)}$,
respectively, where $\delta$ is a sufficiently small positive constant and $\xi_0$ will be related later to the initial condition of \eqref{daet}. 
As they are solutions of \eqref{daet}, they satisfy 
\begin{equation}\label{ini20}
  g(t,w(t),z(t))=0, 
\end{equation}
and
\begin{equation}\label{ini2}
  g(t,\hat{w}(t),\hat{z}(t))=0.
\end{equation}
Denote the partial Jacobian matrices of $\varphi$ and $\psi$ with respect to the second or third argument evaluated at $(t,w_0,z_0(w_0))$ by $\Phi_{w_0}(t),\Phi_{z_0}(t),\Psi_{w_0}(t),\Psi_{z_0}(t)$, respectively. 
By differentiating both sides of \eqref{ini20} with respect to $w_0$, we have 
\begin{equation}\label{ini3}
\begin{aligned}
0&=\tfrac{\partial g}{\partial w}(t,w(t),z(t))\cdot\Big(\Phi_{w_0}(t)+\Phi_{z_0}(t)\cdot z_0'(w_0)\Big)\\&+\tfrac{\partial g}{\partial z}(t,w(t),z(t))\cdot \Big(\Psi_{w_0}(t)+\Psi_{z_0}(t)\cdot z_0'(w_0)\Big).
\end{aligned}
\end{equation}
In the following, we will show that
\begin{equation}\label{varitaj}
\left\{\begin{matrix}
\xi(t):=\lim\limits_{\delta \rightarrow 0} \frac{\varphi(t,w_0+\delta \xi _0,z_0(w_0+\delta \xi_0))-\varphi(t,w_0,z_0(w_0))}{\delta },\\ 
\nu (t):=\lim\limits_{\delta \rightarrow 0} \frac{\psi(t,w_0+\delta \xi _0,z_0(w_0+\delta \xi_0))-\psi(t,w_0,z_0(w_0))}{\delta},
\end{matrix}\right.
\end{equation}
are a pair of solutions of \eqref{vari} with initial value 
\[
\begin{aligned}
\xi(t_0) &= 
\xi_0
,\\ 
\nu (t_0) &= 
\nu_0=\nu_0(\xi_0) :=
\lim\limits_{\delta \rightarrow 0} \tfrac{z_0(w_0+\delta \xi_0)-z_0(w_0)}{\delta}.
\end{aligned}
\]
We can rewrite \eqref{varitaj} as 
\begin{equation}\label{varitaj0}
\left\{\begin{matrix}
\xi(t)=\Big(\Phi_{w_0}(t)+\Phi_{z_0}(t)\cdot z_0'(w_0)\Big)\cdot \xi_0,\\ 
\nu (t)=\Big(\Psi_{w_0}(t)+\Psi_{z_0}(t)\cdot z_0'(w_0)\Big)\cdot\xi_0.
\end{matrix}\right.
\end{equation}
From \eqref{ini3} and \eqref{varitaj0}, we know that \eqref{varitaj} satisfies $0=\tfrac{\partial g}{\partial w}\xi +\frac{\partial g}{\partial z}\nu$.\\
The flow $\varphi (t,w_0,z_0(w_0))$ of \eqref{daet} satisfies
\begin{equation}\label{fl1}
\begin{aligned}
&\varphi (t,w_0,z_0(w_0))=w_0\\&+\int_{0}^{t}f\big(\varphi (\tau,w_0,z_0(w_0)),\psi  (\tau,w_0,z_0(w_0))\big)\dd \tau,
\end{aligned}
\end{equation}
and similarly, the flow $\varphi(t,w_0+\delta\xi_0,z_0(w_0+\delta\xi_0))$ satisfies
\begin{equation}\label{fl2}
\begin{aligned}
&\varphi(t,w_0+\delta\xi_0,z_0(w_0+\delta\xi_0))=w_0+\delta\xi_0+\\&\int_{0}^{t}f\big(\varphi (\tau,w_0+\delta\xi_0,z_0(w_0+\delta\xi_0)),\\&\psi  (\tau,w_0+\delta\xi_0,z_0(w_0+\delta\xi_0))\big)\dd \tau.
\end{aligned}
\end{equation}
Hence,
\begin{equation}\label{varitaj1}
\begin{aligned}
&\xi (t)=\xi_0+\int_{0}^{t}\lim_{\delta \rightarrow 0}\frac{1}{\delta }\times\\&\Big(f\big(\varphi (\tau,w_0+\delta\xi_0,z_0(w_0+\delta\xi_0)),\psi  (\tau,w_0+\delta\xi_0,z_0(w_0+\delta\xi_0))\big)\\&-f\big(\varphi (\tau,w_0,z_0(w_0)),\psi  (\tau,w_0,z_0(w_0))\big)\Big)\dd \tau
\end{aligned}
\end{equation}
Clearly,
\begin{equation}\label{varitaj2}
\begin{aligned}
&\lim_{\delta \rightarrow 0}\frac{1}{\delta }\times\\&\Big(f\big(\varphi (\tau,w_0+\delta\xi_0,z_0(w_0+\delta\xi_0)),\psi  (\tau,w_0+\delta\xi_0,z_0(w_0+\delta\xi_0))\big)\\
&\quad -f\big(\varphi (\tau,w_0,z_0(w_0)),\psi  (\tau,w_0,z_0(w_0))\big)\Big)
\\&=\tfrac{\partial f}{\partial w}(\tau,w(\tau),z(\tau))\cdot \Big(\Phi_{w_0}(\tau)+\Phi_{z_0}(\tau)\cdot z_0'(w_0)\Big)\cdot \xi_0\\
&\quad +\tfrac{\partial f}{\partial z}(\tau,w(\tau),z(\tau))\Big(\Psi_{w_0}(\tau)+\Psi_{z_0}(\tau)\cdot z_0'(w_0)\Big)\cdot \xi_0\\
&\overset{\eqref{varitaj0}}{=} \tfrac{\partial f}{\partial w}(\tau,w(\tau),z(\tau))\cdot \xi(\tau) + \tfrac{\partial f}{\partial z}(\tau,w(\tau),z(\tau)) \cdot \nu(\tau).
\end{aligned}
\end{equation}

Substituting this back to \eqref{varitaj1} and differentiating with respect to time gives us
\begin{equation}\label{varitaj8}
\begin{aligned}
\dot{\xi }(t)=\frac{\partial f}{\partial w}\xi(t) +\frac{\partial f}{\partial z}\nu(t).
\end{aligned}
\end{equation}
Altogether this shows that indeed $\xi(t)$, $\nu(t)$ given by \eqref{varitaj} is a solution of \eqref{vari}. 

We can now show the sufficiency result. 

\emph{Contracting $\Rightarrow$ UGES.}
Let $c$ and $\alpha$ be the constants corresponding to the contractivity condition. 
Seeking a contradiction, assume the variational DAE system \eqref{vari} is not UGES. Then 
there exists a solution $\sbm{w(\cdot) \\z(\cdot)}$ of \eqref{daet} and an initial value $\sbm{w_0 \\z_0(w_0)}$ such that for the corresponding solution $\sbm{\xi(\cdot) \\\nu(\cdot)}$ of \eqref{vari} we have that 
for  $c':=\frac{3}{2}c$ and $\alpha':=\alpha$, 
there exists $T>0$ such that 
\begin{equation} \label{pf90}
\left\|\bbm{\xi(T) \\\nu(T)}\right\| > c'e^{-\alpha' T} \left\|\bbm{\xi_0 \\\nu_0}\right\| = \frac{3}{2}ce^{-\alpha T} \left\|\bbm{\xi_0 \\\nu_0}\right\| \qquad \forall \xi_0.
\end{equation}  
Let $\sbm{\hat{w}(\cdot) \\\hat{z}(\cdot)}$ be a solution of \eqref{daet} with initial value $\sbm{\hat{w}(t_0) \\\hat{z}(t_0)}
=
\sbm{w_0+ \delta \xi_0 \\z_0(w_0+ \delta \xi_0)}$ for sufficiently small 
$\delta\in\mathbb{R}$. By definition, $
\xi(t)= \lim_{\delta \rightarrow 0}\frac{\hat{w}(t)-w(t)}{\delta }$ and $
\nu (t)=\lim_{\delta \rightarrow 0}\frac{\hat{z}(t)-z(t)}{\delta }
$;
hence, for a sufficiently small $\delta>0$, we have that at time $T$,  
\begin{equation} \label{pf72}
\begin{aligned}
 \frac{\left\|
\bbm{\hat{w}(T) \\\hat{z}(T)}-\bbm{w(T) \\z(T)}
\right\|}{\delta }>\frac{4}{5}\left\|\bbm{\xi(T) \\\nu(T)}
\right\|,
\end{aligned}
\end{equation}
where the lower-bound constant $\frac{4}{5}$ is chosen arbitrarily for the following computation of bounds. 
Similarly, since $\nu_0:=\lim_{\delta\to 0} \frac{\hat{z}(t_0)-z(t_0)}{\delta}$, for a sufficiently small $\delta>0$, we have that
\begin{equation} \label{pf720}
\begin{aligned}
 \left\|\nu_0\right\|>\frac{5}{6} \frac{\left\|\hat{z}(t_0)-z(t_0)\right\|}{\delta}.
\end{aligned}
\end{equation}
Combining \eqref{pf90}, \eqref{pf72} and \eqref{pf720}, we obtain 
\begin{align} 
 \left\|
\bbm{\hat{w}(T) \\\hat{z}(T)}  -\bbm{w(T) \\z(T)}
\right\| &\overset{\eqref{pf72}}{>}\frac{4}{5}\delta\left\|\bbm{\xi(T) \\\nu(T)}\right\|\overset{\eqref{pf90}}{>} \frac{6}{5}c e^{-\alpha T}\left\| 
\bbm{\delta \xi_0 \\\delta \nu_0}
\right\|\nonumber\\&\overset{\eqref{pf720}}{>}ce^{-\alpha T}\left\|
\bbm{\frac{6}{5}\hat{w}(t_0) \\\hat{z}(t_0)}-\bbm{\frac{6}{5}w(t_0) \\z(t_0)}
\right\|\nonumber\\&>ce^{-\alpha T}\left\|
\bbm{\hat{w}(t_0) \\\hat{z}(t_0)}-\bbm{w(t_0) \\z(t_0)}
\right\|
\nonumber,
\end{align}
for all $\xi_0$, the last inequality arises from the property of the $p$-norm.
This is in contradiction to the contractivity of \eqref{daet} and concludes the proof of the sufficiency part. 

\emph{UGES $\Rightarrow$ Contracting.} Let us consider two solutions $\sbm{
w(\cdot)\\ 
z(\cdot)
} = \sbm{
\varphi(\cdot,w_0,z_0(w_0))\\ 
\psi(\cdot,w_0,z_0(w_0))
}$ and $\sbm{
\hat{w}(\cdot)\\ 
\hat{z}(\cdot)
} = \sbm{
\varphi(\cdot,\hat{w}_0,\hat{z}_0(\hat{w}_0))\\ 
\psi(\cdot,\hat{w}_0,\hat{z}_0(\hat{w}_0))
}$ of \eqref{daet}. 
Consequently, we can utilize the fundamental theorem of calculus for line integrals to obtain
\begin{equation}\label{pf13}
  \left\{\begin{matrix}
 \hat{w}(t) - w(t) = \int_{w_0}^{\hat{w}_0} \frac{\dd\varphi(t,\zeta,z_0(\zeta) )}{\dd \zeta} \dd \zeta\\ 
\hat{z}(t) - z(t)=\int_{w_0}^{\hat{w}_0} \frac{\dd\psi(t,\zeta,z_0(\zeta))}{\dd \zeta}\dd \zeta 
\end{matrix}.\right.
\end{equation}
According to \eqref{varitaj}, one has
\begin{equation}\label{pf131}
\bbm{\xi(t)\\ 
\nu(t)}=
\bbm{\frac{\dd\varphi(t,w_0,z_0(w_0) )}{\dd w_0} &0 \\ 
\frac{\dd\psi(t,w_0,z_0(w_0) )}{\dd w_0}&0
 }
\bbm{\xi_0\\\nu_0}.
\end{equation}
Then, in \eqref{pf131}, $
\bbm{\frac{\dd\varphi(t,w_0,z_0(w_0) )}{\dd w_0} &0 \\ 
\frac{\dd\psi(t,w_0,z_0(w_0) )}{\dd w_0}&0
 }$ 
is the (singular) state transition matrix of the DAE \eqref{vari}. 
From the UGES property of \eqref{vari}, it follows with similar arguments as in the necessity proof of  \cite[Thm.~4.11]{Khalil2002} that
\begin{equation} \label{pf14}
\begin{aligned}
&\left\|\bbm{\frac{\dd\varphi(t,w_0,z_0(w_0) )}{\dd w_0} &0 \\ 
\frac{\dd\psi(t,w_0,z_0(w_0) )}{\dd w_0}&0
 }
\right\|\leq ce^{-\alpha t}, 
\end{aligned}
\end{equation}
from which it follows\footnote{
In general, for any $p$-norm and its induced matrix norm we have $\left\|\bbm{M&0}\right\|=\underset{\left\|\sbm{x\\y} \right\|=1}{\sup}\left\|\bbm{M&0}\bbm{x\\y}\right\|=\underset{\left\|\sbm{x\\y} \right\|=1}{\sup}\left\|Mx\right\|\geq \underset{\left\|x \right\|=1}{\sup}\left\|Mx\right\|=\left\|M\right\|$. The inequality follows from the set inclusion $\{x|\left\|x\right\|=1\}=\{x|\left\|\sbm{x\\0}\right\|=1\}\subseteq \{x|\exists y:\left\|\sbm{x\\y}\right\|=1\}$.\\
} that 
\begin{equation} \label{pf15}
\begin{aligned}
\left\|\bbm{\frac{\dd\varphi(t,w_0,z_0(w_0) )}{\dd w_0}\\ 
\frac{\dd\psi(t,w_0,z_0(w_0) )}{\dd w_0}}
\right\|\leq ce^{-\alpha t}
\end{aligned}
\end{equation}

Using \eqref{pf15} to get the upper bound of \eqref{pf13}, we have
\begin{equation} \label{pf16}
\begin{aligned}
\left\|
\bbm{\hat{w}(t)\\ 
\hat{z}(t)}
- \bbm{w(t)\\ 
z(t)}
\right\|&=\left\|
\int_{w_0}^{\hat{w}_0} \bbm{\frac{\dd\varphi(t,\zeta,z_0(\zeta) )}{\dd \zeta}\\ 
\frac{\dd\psi(t,\zeta,z_0(\zeta) )}{\dd \zeta}}\dd \zeta\right\|
\\&
\leq  ce^{-\alpha t}\left\|
\hat{w}_0 -
w_0
\right\|
\\&
\leq  ce^{-\alpha t}\left\|
\bbm{\hat{w}_0\\ 
z_0(\hat{w}_0)} -
\bbm{w_0\\ 
z_0(w_0)}
\right\|,
\end{aligned}
\end{equation}
where the latter inequality follows from $\|x\|=\| (\begin{smallmatrix} x\\ 0 \end{smallmatrix})\| \leq \|(\begin{smallmatrix} x\\ y \end{smallmatrix})\|$, which is true for any $p$-norm. This shows that \eqref{daet} is contracting and the proof is complete.
\end{proof}

Proposition \ref{proposition1} shows that the contractivity of DAE systems inherit the property of contracting ODE systems (i.e., \cite[Prop.~1]{Barabanov2019}) on a corresponding manifold $\tfrac{\partial g}{\partial w}(t,w,z)\xi +\frac{\partial g}{\partial z}(t,w,z)\nu=0$. 
From \eqref{pf131} and \eqref{pf15}, we can deduce that $\|\xi(t)\|\leq ce^{-\alpha t}\|\xi_0\|$, which means that when the system \eqref{vari} is exponentially stable, its reduced system \eqref{eq:re} is exponentially stable. In the case of a time-invariant system, 
the reverse implication holds automatically. This is due to the fact that it becomes feasible to guarantee the boundedness of $\left[\frac{\partial g}{\partial z}\right]^{-1}\tfrac{\partial g}{\partial w}$ within a particular invariant set. As a result, we can utilize Proposition \ref{proposition1} for analyzing the stability of time-invariant DAE systems, as demonstrated in Corollary \ref{corollary1}. Apart from the time-invariant system, the presence of an invariant set might not be guaranteed for a time-varying system, i.e., $\left[\frac{\partial g}{\partial z}\right]^{-1}\tfrac{\partial g}{\partial w}$ could potentially become unbounded. This implies that the contractivity of the time-varying DAE system cannot be derived from the contractivity of the reduced system in the absence of bounded condition on $\left[\frac{\partial g}{\partial z}\right]^{-1}\tfrac{\partial g}{\partial w}$. 
In Lemma \ref{lemma1} below we relax the condition on $\left[\frac{\partial g}{\partial z}\right]^{-1}\tfrac{\partial g}{\partial w}$ by analysing the stability of an auxiliary ODE system instead of analysing the contractivity of the reduced system. By adopting this approach, it becomes unnecessary to have prior knowledge about the reduced system \eqref{eq:re}, as demonstrated in Example \ref{ex1}.

\subsection{The ODE approach}\label{sectionode}
In this section, we present an ODE approach to analyze the contractivity of nonlinear time-varying DAE systems. 
We present a construction of auxiliary ODE whose trajectories can represent the convergence property of the variational DAE system. As a consequence, 
we can apply traditional control theories to the auxiliary ODE in order to analyze the UGES of \eqref{vari}.

For a given variational DAE system \eqref{vari}, we construct its auxiliary ODE system by 
\begin{equation}\label{auxode}
\begin{aligned}
 \left\{\begin{matrix}
\dot{\xi}_\gamma=\frac{\partial f}{\partial w}\xi_\gamma+\frac{\partial f}{\partial z}\nu_\gamma,\\ 
\dot{\nu}_\gamma=-\left[\frac{\partial g}{\partial z}\right]^{-1}\Big[\Big(\gamma\tfrac{\partial g}{\partial w}+\frac{\dd}{\dd t}(\tfrac{\partial g}{\partial w})+\tfrac{\partial g}{\partial w}\frac{\partial f}{\partial w}\Big)\xi_\gamma\\\qquad+\Big(\gamma\frac{\partial g}{\partial z}+\frac{\dd}{\dd t}(\frac{\partial g}{\partial z})+\tfrac{\partial g}{\partial w}\frac{\partial f}{\partial z}\Big)\nu_\gamma\Big],
\end{matrix}\right.
\end{aligned}   
\end{equation}
where $\gamma$ is a no-negative constant, $\frac{\dd}{\dd t}(\tfrac{\partial g}{\partial w})=\frac{\partial^2 g}{\partial w \partial t}+\frac{\partial^2 g}{\partial^2 w}f-\frac{\partial^2 g}{\partial w \partial z}\left[\frac{\partial g}{\partial z}\right]^{-1}(\frac{\partial g}{\partial t}+\tfrac{\partial g}{\partial w}f)$, and $\frac{\dd}{\dd t}(\frac{\partial g}{\partial z})=\frac{\partial^2 g}{\partial z \partial t}+\frac{\partial^2 g}{\partial w \partial z}f+\frac{\partial^2 g}{\partial^2 z}\left[\frac{\partial g}{\partial z}\right]^{-1}(\frac{\partial g}{\partial t}-\tfrac{\partial g}{\partial w}f)$. It is worth noting that in various control problems that involve the system's output, such as state observer design and output feedback control, the output of the system, denoted as $z=h(w,t)$, can be considered as a time-varying constraint. In other words, the function $g(t,w,z)$ takes the form of $z-h(w,t)$. In such cases, the term $\frac{\dd}{\dd t}(\frac{\partial g}{\partial z})$ becomes zero, and $\frac{\dd}{\dd t}(\tfrac{\partial g}{\partial w})$ can be simplified to $\frac{\partial^2 g}{\partial w \partial t}+\frac{\partial^2 g}{\partial^2 w}f$.
\begin{lemma}\label{lemma1}
The variational DAE system \eqref{vari} is exponentially
stable if 
there exist $\gamma\geq 0$ such that 
its auxiliary ODE system \eqref{auxode} is exponentially
stable.
\end{lemma}
\begin{proof}
Consider an auxiliary system of \eqref{vari} as follows, which is equivalent to \eqref{auxode},
\begin{equation}\label{auxode1}
 \left\{\begin{matrix}
\dot{\xi}_\gamma=\frac{\partial f}{\partial w}\xi_\gamma+\frac{\partial f}{\partial z}\nu_\gamma,\\ 
\dot{\overbrace{\Big(\tfrac{\partial g}{\partial w}\xi_\gamma+\frac{\partial g}{\partial z}\nu_\gamma\Big)}}=-\gamma\Big(\tfrac{\partial g}{\partial w}\xi_\gamma+\frac{\partial g}{\partial z}\nu_\gamma\Big).
\end{matrix}\right.   
\end{equation}
From the second equation in \eqref{auxode1}, we have
\begin{equation}\label{traj}
\Big(\tfrac{\partial g}{\partial w}\xi_\gamma+\frac{\partial g}{\partial z}\nu_\gamma\Big)=\Big(\frac{\partial g_0}{\partial w}\xi_{\gamma0}+\frac{\partial g_0}{\partial z}\nu_{\gamma0}\Big)e^{-\gamma t}. 
\end{equation}
where $\frac{\partial g_0}{\partial w}=\tfrac{\partial g}{\partial w}\Big|_{t=0}$.  It follows then that the solution of \eqref{vari} is a particular case of \eqref{auxode1} with initial condition $\frac{\partial g_0}{\partial w}\xi_{\gamma0}+\frac{\partial g_0}{\partial z}\nu_{\gamma0}=0$. 
\end{proof}

In Lemma \ref{lemma1}, we use the stability properties of \eqref{auxode1} to derive the stability properties of \eqref{vari}. In this regard, the constant $\gamma$ must be chosen properly 
as \eqref{auxode1} may fail to capture the stability properties of \eqref{vari}, e.g., the system \eqref{auxode1} can be unstable while correspondingly the system \eqref{vari} is stable. For instance, 
let us consider the following contracting DAE system whose trajectories can be calculated explicitly
\begin{equation}\label{exam3}
\left\{\begin{aligned}
\dot{w}&=-2e^{t}z,\\ 
0&=e^{-t}w-z,
\end{aligned}\right.   
\end{equation}
which again has the same form as its variational system. The trajectories of the system are given by  $w(t)=w_0e^{-2t}$ and $z(t)=w_0e^{-3t}$. If we choose $\gamma=1$, \eqref{auxode} can be rewritten as
\begin{equation}\label{exam30}
\left\{\begin{aligned}
\dot{w}&=-2e^{t}z,\\ 
\dot{z}&=-3z,
\end{aligned}\right.   
\end{equation} 
which has the solution $w(t)=z_0e^{-2t}+w_0-z_0$ and $z(t)=z_0e^{-3t}$ which does not converge to zero for $w_0\neq z_0$. However, by choosing $\gamma=4$, \eqref{auxode} can be rewritten as
\begin{equation}\label{exam31}
\left\{\begin{aligned}
\dot{w}&=-2e^tz,\\ 
\dot{z}&=3e^{-t}w-6z,
\end{aligned}\right.   
\end{equation} 
whose solution is given by $w(t)=(3w_0-2z_0)e^{-2t}-2(w_0-z_0)e^{-3t}$ and $z(t)=(3w_0-2z_0)e^{-3t}-3(w_0-z_0)e^{-4t}$, which is converging to zero. Based on this observation, it appears that in order to ensure the contractivity of the system, the parameter $\gamma$ should be chosen appropriately.

In the following proposition, 
we establish a lower bound for $\gamma$ based on a priori information of the DAE system. 

\begin{proposition}\label{proposition2}
Suppose that the reduced system \eqref{eq:re} is uniformly 
exponentially upper and lower bounded by $-\underline{\alpha}$, $-\overline{\alpha}$, with $\overline{\alpha}\geq \underline{\alpha}>0$, i.e.\ there exists constants $\underline{c},\overline{c}>0$ such that all solutions satisfy
\[
   \underline{c} e^{-\overline{\alpha} (t-t_0)} \left\|\xi(t_0)\right\| \leq \left\|\xi(t)\right\| \leq \overline{c} e^{-\underline{\alpha} (t-t_0)}  \left\|\xi(t_0)\right\|
\]
for all $t\geq t_0$. Furthermore, assume that there exist constants $l_f = l_f(w_0,z_0)$, $ k_f = k_f(w_0,z_0)$,  $l_g = l_g(w_0,z_0)$, $k_g=k_g(w_0,z_0)$ with $k_f,k_g>0$ such that, $\left\|\frac{\partial f}{\partial z}\right\|\leq k_fe^{l_ft}$ and $\left\|\left[\frac{\partial g}{\partial z}\right]^{-1}\right\|\leq k_ge^{l_gt}$ for all $t\geq 0$, where all gradients are evaluated on the trajectory $w(t)=\varphi(t,w_0,z_0)$ and $z(t)=\psi(t,w_0,z_0)$ of \eqref{daet}; let $l := \max\{l_g, l_g+l_f\}$. If the variational DAE \eqref{vari} is exponentially stable and $\gamma > l+\overline{\alpha}$ then also the corresponding ODE system \eqref{auxode} is exponentially stable (with a convergence rate at least as large as any convergence rate of the variational DAE). In particular, if for sufficiently large $\gamma$ the ODE \eqref{auxode} is unstable, the variational DAE cannot be exponentially stable and hence the original nonlinear DAE cannot be contractive. 
\end{proposition}
\begin{proof}
Let us define $q=\tfrac{\partial g}{\partial w}\xi_\gamma+\frac{\partial g}{\partial z}\nu_\gamma$, so that \eqref{auxode1} can be rewritten as
\begin{equation}\label{auxode10}
 \left\{\begin{matrix}
\dot{\xi}_{\gamma}=\Big(\frac{\partial f}{\partial w}-\frac{\partial f}{\partial z}\left[\frac{\partial g}{\partial z}\right]^{-1}\tfrac{\partial g}{\partial w}\Big)\xi_{\gamma}+\frac{\partial f}{\partial z}\left[\frac{\partial g}{\partial z}\right]^{-1}q,\\ 
\dot{q}=-\gamma q.
\end{matrix}\right.   
\end{equation}
In this case, the solution of \eqref{auxode10} is given by 
\begin{equation}\label{auxode11}
\left\{\begin{matrix}
\xi_{\gamma}=\Phi(t,t_0)\xi_0+q_0\int_{t_0}^{t}{\Phi(t,s)\Big(\frac{\partial f}{\partial z}\left[\frac{\partial g}{\partial z}\right]^{-1}e^{-\gamma s}\Big)\dd s},\\ 
q=q_0e^{-\gamma t},
\end{matrix}\right.   
\end{equation}
where $\Phi(t,s)$ is the transition matrix of the reduced system. By assumption, the transition matrix satisfies $\underline{c}e^{-\overline{\alpha}(t_1-t_0)}\leq||\Phi(t_1,t_0)||\leq \overline{c}e^{-\underline{\alpha}(t_1-t_0)}$, hence we arrive at the following estimation 
\begin{equation}\label{auxode12}
\begin{aligned}
\left\|\xi_{\gamma}(t)\right\| & \leq   \overline{c}e^{-\underline{\alpha} (t-t_0)}\left\|\xi_0\right\|\\
&\quad +\left\|q_0\right\| \int_{t_0}^{t}\overline{c} e^{-\underline{\alpha}(t-s)}\left\|\tfrac{\partial f}{\partial z}\right\| \left\|\left[\tfrac{\partial g}{\partial z}\right]^{-1}\right\|e^{-\gamma s}\dd s\\
&\leq \overline{c} e^{-\underline{\alpha} t}\left\|\xi_0\right\|\\
&\quad + \overline{c} e^{-\underline{\alpha} t}\left\|q_0\right\|\int_{t_0}^{t} 
k_f e^{l_f s} k_g e^{l_g s}
e^{-(\gamma-\underline{\alpha}) s}\dd s\\
& \leq \overline{c} e^{-\underline{\alpha} (t-t_0)}(\left\|\xi_0\right\|+k\left\|q_0\right\|).
\end{aligned}  
\end{equation}
The last inequality is derived 
from $l_f+l_g-\gamma+\underline{\alpha}\leq l -\gamma + \overline{\alpha} < 0$ and for some suitable $k>0$. 
This implies that the convergence rate of $\xi$ is at least $\underline{\alpha}$.
By the definition of $q$ defined before \eqref{auxode10}, we have $\nu_\gamma = -\left[\frac{\partial g}{\partial z}\right]^{-1}\left(e^{-\gamma t}q_0 - \tfrac{\partial g}{\partial w} \xi_\gamma\right)$ and hence
\begin{equation}\label{auxode13}
\begin{aligned}
&\left\|\nu_{\gamma}\right\|\leq \left\|q_0\right\|\left\|\left[\tfrac{\partial g}{\partial z}\right]^{-1}\right\|e^{-\gamma t}+\left\|\left[\tfrac{\partial g}{\partial z}\right]^{-1}\tfrac{\partial g}{\partial w}\xi_{\gamma}\right\|
\end{aligned}  
\end{equation}
By assumption, $l_g - \gamma < - \overline{\alpha}$, and hence
\[
 \left\|q_0\right\|\left\|\left[\tfrac{\partial g}{\partial z}\right]^{-1}\right\|e^{-\gamma t} \leq k_g \|q_0\| e^{-\overline{\alpha} t}.
\]
We furthermore have, utilizing \eqref{auxode11} and $\Phi(t,s) = \Phi(t,t_0)\Phi(t_0,s)$,
\begin{equation}\label{auxode14}
\begin{aligned}
&\left\|\left[\tfrac{\partial g}{\partial z}\right]^{-1}\tfrac{\partial g}{\partial w}\xi_{\gamma}\right\| \leq \left\|\left[\tfrac{\partial g}{\partial z}\right]^{-1}\tfrac{\partial g}{\partial w}\Phi(t,t_0)\right\| \Bigg ( \|\xi_0\|+\\
&
\|q_0\| \int_{t_0}^t \left\|\Phi(t_0,s)\tfrac{\partial f}{\partial z}\left[\tfrac{\partial g}{\partial z}\right]^{-1}e^{-\gamma s}\right\|\dd s \Bigg)
\end{aligned}  
\end{equation}
To further bound the terms in \eqref{auxode14}, we first observe that the solution of \eqref{vari} satisfies 
\begin{equation}\nonumber
\begin{aligned}
\begin{pmatrix}\xi(t)\\\nu(t)\end{pmatrix} &= \begin{bmatrix} \Phi(t,t_0) & 0 \\ \left[\tfrac{\partial g}{\partial z}\right]^{-1}\tfrac{\partial g}{\partial w}\Phi(t,t_0)&0\end{bmatrix} \begin{pmatrix} \xi_0\\ \nu_0\end{pmatrix}
\end{aligned}  
\end{equation}
hence exponential stability of \eqref{vari} ensures
\begin{equation}\nonumber
\begin{aligned}
\left\|\bbm{\Phi(t,t_0)&0\\\left[\tfrac{\partial g}{\partial z}\right]^{-1}\tfrac{\partial g}{\partial w}\Phi(t,t_0)&0}\right\|\leq c'e^{-\alpha' t},
\end{aligned}  
\end{equation}
for some positive $c'$ and $\alpha'\leq \underline{\alpha}$. It follows \footnote{In general, for any induced matrix norm by a $p$-norm we have 
 $\left\|\bbm{M & 0 \\N & 0 }\right\|\geq \left\|N\right\|$.}  that 
\begin{equation*}
\begin{aligned}
\left\|\left[\tfrac{\partial g}{\partial z}\right]^{-1}\tfrac{\partial g}{\partial w}\Phi(t,t_0)\right\|\leq c'e^{-\alpha't},
\end{aligned}  
\end{equation*}
To bound the integral term in \eqref{auxode14} we observe that
\[
   \|\Phi(t_0,s)\| \leq \|\Phi(s,t_0)\|^{-1} \leq \underline{c} e^{\overline{\alpha}(s-t_0)}
\]
hence
\[\begin{aligned}
   \left\|\Phi(t_0,s)\tfrac{\partial f}{\partial z}\left[\tfrac{\partial g}{\partial z}\right]^{-1}e^{-\gamma s}\right\|&\leq \underline{c} e^{\overline{\alpha}(s-t_0)} k_f e^{l_f s} k_g e^{l_g s} e^{-\gamma s}\\
   &\leq \tilde{c} e^{(\overline{\alpha}+l_f+l_g - \gamma)s} = \tilde{c} e^{-\tilde{\alpha} s}
   \end{aligned}
\]
for some $\tilde{c}>0$ and $\tilde{\alpha}:=-(\overline{\alpha}+l_f+l_g - \gamma) > 0$. Consequently, the integral term in \eqref{auxode14} is bounded by, say, $\hat{c}>0$  (independent of $t$). Altogether, we can bound $\nu_\gamma$ by
\[
  \|\nu_\gamma\| \leq k_g \|q_0\| e^{-\overline{\alpha} t} + c' e^{-\alpha' t}(\|\xi_0\| + \hat{c} \|q_0\|).
\]

Together with \eqref{auxode12} this shows that the system \eqref{auxode} is asymptotically stable. 
Note that $-\overline{\alpha} < -\underline{\alpha} \leq -\alpha'$, hence the overall exponential convergence rate of \eqref{auxode12} is at least $\alpha'$ which was the assumed convergence rate of \eqref{vari}.
\end{proof} 


In Proposition \ref{proposition2}, we have a mild assumption where the time varying functions $\frac{\partial f}{\partial z}$ and $\left[\frac{\partial g}{\partial z}\right]^{-1}$ are bounded by $ke^{lt}$ instead of by constants. 
Let us recall again the previous example in \eqref{exam3}, where the reduced system is $\dot{w}=-2w$, we have $\left\|\frac{\partial f}{\partial z}\right\|=2e^{t}$, $\left\|\left[\frac{\partial g}{\partial z}\right]^{-1}\right\|=1$. Thus the hypotheses in  Proposition \ref{proposition2} hold with $\alpha =2$, $k=l=1$ and $\gamma=4> l+\overline{\alpha}=3$. Then, the ODE system \eqref{exam31} has the same exponential convergence rate as the corresponding DAE system \eqref{exam30}, namely $\alpha' = 2$. By considering Proposition \ref{proposition2}, it becomes evident that selecting $\gamma > l + \overline{\alpha}$ allows us to conclude that if \eqref{auxode} is unstable, then \eqref{vari} is also unstable. This proposition helps prevent any potential ``misjudgment" similar to the previous example \eqref{exam30}. By examining \eqref{auxode13}, it shows 
that as $\gamma$ increases, the estimation performance of \eqref{auxode1} w.r.t. \eqref{vari} improves. 
If obtaining prior information about the DAE system \eqref{eq:d1} is difficult, one can select a sufficiently large value for $\gamma$ to apply Proposition \ref{proposition2} and determine the contracting rate of the DAE system.

By utilizing the concept of matrix measure, we will examine the contractivity of the time-varying DAE system \eqref{daet} by focusing on the generalized Jacobian matrix of the corresponding auxiliary ODE system \eqref{auxode}. In order to simplify the notation, we rewrite \eqref{auxode} into 
\begin{equation}\label{auxode0}
\begin{aligned}
 \left\{\begin{matrix}
\dot{\xi }_\gamma=A(t)\xi_\gamma +B(t)\nu_\gamma,\\ 
\dot{\nu}_\gamma=-F^{-1}(t)C(t)\xi_\gamma-F^{-1}(t)D(t)\nu_\gamma,
\end{matrix}\right.
\end{aligned}   
\end{equation}
where $A(t)=\frac{\partial f}{\partial w}$, $B(t)=\frac{\partial f}{\partial z}$, $F(t)=\frac{\partial g}{\partial z}$, $C(t)=\gamma\tfrac{\partial g}{\partial w}+\frac{\dd}{\dd t}(\tfrac{\partial g}{\partial w})+\tfrac{\partial g}{\partial w}\frac{\partial f}{\partial w}$, and $D(t)=\gamma\frac{\partial g}{\partial z}+\frac{\dd}{\dd t}(\frac{\partial g}{\partial z})+\tfrac{\partial g}{\partial w}\frac{\partial f}{\partial z}$. Furthermore, for any sufficiently smooth invertible matrix  $M(w,z,t)$ we define the matrix 
\begin{equation}\label{j1}
J_M(t):=\dot{M}M^{-1}+M\bbm{
A & B\\ 
-F^{-1}C & -F^{-1}D
}M^{-1},
\end{equation}
where $M(w(t),z(t),t)$ is evaluated along a solution $(w,z)$ of \eqref{daet} and $\dot{M}(w(t),z(t),t)=\tfrac{\mathrm{d}}{\mathrm{d}t} M(w(t),z(t),t)$.

\begin{theorem}\label{theorem1}
If there exist an invertible matrix $M(w,z,t)$ and $\beta>0$ such that for all solutions $(w,z)$ of \eqref{daet} it holds that $\left\|M^{-1}\right\|\left\| M \right\|$ is bounded and $\mu_q\Big(J_M(t)\Big)\leq -\beta$, where $J_M(t)$ is given by \eqref{j1}, then the DAE system \eqref{daet} is contracting.
\end{theorem}

\begin{proof}
For any $M(w,z,t)$, 
we consider the coordinate transformations $\sbm{p\\r}=M\sbm{\xi_\gamma\\\nu_\gamma}$, so that \eqref{auxode0} can be rewritten as
\begin{equation}\label{auxode01}
\begin{aligned}
 \dot{\overbrace{\bbm{
p\\ 
r
}}}=\underset{J_M(t)}{\underbrace{\Big(\dot{M}M^{-1}+M\bbm{
A & B\\ 
-F^{-1}C & -F^{-1}D
}M^{-1}\Big)}}\bbm{
p\\ 
r
}.
\end{aligned}   
\end{equation}
Appealing to Coppel’s inequality \cite{Vidyasagar}, we can conclude that
\begin{equation}\label{dini2}
\begin{aligned}
\left\|\bbm{
p(t)\\ 
r(t)
}\right\|\leq e^{-\beta t}\left\|\bbm{
p(0)\\ 
r(0)
}\right\|,
\end{aligned}
\end{equation}
from which it follows that 
\begin{equation}\label{dini21}
\begin{aligned}
&\left\|\bbm{
\xi_\gamma(t)\\ 
\nu_\gamma(t)
}\right\|=\left\|M ^{-1}
\bbm{
p(t)\\ 
r(t)
}\right\|
\overset{\eqref{dini2}}{\leq} e^{-\beta t}\left\|M ^{-1}
\right\|\left\|M\right\|\left\|\bbm{
\xi_\gamma(0)\\ 
\nu_\gamma(0)
}\right\|\\&\leq ce^{-\beta t}\left\|\bbm{
\xi_\gamma(0)\\ 
\nu_\gamma(0)
}\right\|,
\end{aligned}
\end{equation}
i.e., the system \eqref{auxode0} is exponentially stable, the last inequality is due to $\left\|M^{-1}\right\|\left\|M\right\|$ is bounded. By Lemma \ref{lemma1} and Proposition \ref{proposition1}, we can conclude that the DAE system \eqref{daet} is contracting.\end{proof}

In Theorem \ref{theorem1}, the condition that $\left\|M^{-1}\right\|\left\|M\right\|$ is bounded is less conservative compared to the usual assumption of being a Lyapunov transformation \cite[Def.~2.4]{dieci2002lyapunov}. Specifically, instead of requiring both $\left\|M^{-1}\right\|$ and $\left\|M\right\|$ to be bounded, it only demands their product to be bounded. As illustrated in Example \ref{ex10}, this less conservative condition allows for cases where $M=\sbm{
e^{-t}&0&0\\
0&e^{-t}&0\\
0&0&e^{-t}
}$, which has an unbounded inverse. Nevertheless, it is trivial to see that $\left\|\sbm{
e^{t}&0&0\\
0&e^{t}&0\\
0&0&e^{t}}\right\|\left\|\sbm{e^{-t}&0&0\\
0&e^{-t}&0\\
0&0&e^{-t}}\right\|=1$. 

\begin{remark}\label{remark2}
As a powerful tool to analyze the stability of ODE systems, the Lyapunov method has been widely used in the analysis and control design of nonlinear systems. 
By employing the 2-norm matrix measure, the 
hypotheses in Theorem \ref{theorem1}, i.e., $\mu_2\Big(J_M(t)\Big)\leq-\beta$ can be  
reformulated into solvability of the following differential Riccati inequality  
\begin{equation}\label{ric}
\begin{aligned}
&\bbm{
A^{\top}(t) & C^{\top}(t)F^{-\top}(t)\\ 
B^{\top}(t) & D^{\top}(t)F^{-\top}(t)
}P(t)+\\&P(t)\bbm{
A(t) &  B(t)\\
F^{-1}(t)C(t) & F^{-1}(t)D(t)
}+\dot{P}(t)\leq -\beta P(t),
\end{aligned}
\end{equation}
where $P(t)=M^{\top}M$ for some $M(w,z,t)$. Afterwards, Theorem \ref{theorem1} can be simplified to \cite[Thm.~2]{Lohmiller1998}. As will be shown later in Example \ref{ex10}, 
such Riccati inequality will be used for the design of an observer. 
Nevertheless, in certain scenarios, it can be challenging to find a solution $M^{\top}M$ that satisfies \eqref{ric}. By employing the 1-norm ($\infty$-norm) as matrix measure, we can arrive at a simple numerical test to each column (or row) 
of the matrix $J_M(t)$. We demonstrate this approach via Example \ref{ex1} and Example \ref{ex2} later.
\end{remark}


\section{Simulation Setup and Applications}
This section presents three numerical examples to illustrate different applications of the proposed methods. In the first example, we examine the contraction property of a nonlinear time-varying DAE system using Theorem \ref{theorem1}. The second example illustrates the practical application of Proposition \ref{proposition1} in stabilizing a time-invariant DAE system, as demonstrated by Corollary \ref{corollary1}. Subsequently, we employ Corollary \ref{corollary1} to design a state feedback controller for a power source system with an inverter interface, which can be modeled as a time-invariant DAE system. Finally, the third illustration involves the application of Theorem \ref{theorem1} to devise an observer for a time-varying ODE system, namely, Corollary \ref{corollary0}. Subsequently, we employ Corollary \ref{corollary0} to design an observer for an unstable time-varying ODE system.

\subsection{Contractivity of nonlinear time-varying DAE systems}
\begin{example}\label{ex1}
Consider a nonlinear time-varying DAE system
\begin{equation}\label{smex1}
\begin{aligned}
\left\{\begin{matrix}
\dot{w}_1=-4w_1-0.5\cos z,\\ 
\dot{w}_2=\frac{4}{3+\sin t}w_1-\frac{3+\cos t}{3+\sin t}w_2-\frac{4}{3+\sin t},\\
0=4z+0.5\sin z+w_1+(3+\sin t)w_2.
\end{matrix}\right.
\end{aligned}
\end{equation}
where $w(t)\in\mathbb{R}^{2}$ is the state vector and $z(t)\in\mathbb{R}$ refers to the
algebraic vector. Its variational system is 
\begin{equation}\label{smex10}
\begin{aligned}
\left\{\begin{smallmatrix}
\dot{\xi}_1=-4\xi_1+(0.5\sin z) \nu,\\ 
\dot{\xi}_2=\frac{4}{3+\sin t}\xi_1-\frac{3+\cos t}{3+\sin t}\xi_2,\\ 
0=\xi_1+(3+\sin t)\xi_2+(4+0.5\cos z)\nu.
\end{smallmatrix}\right.
\end{aligned}
\end{equation}
In the given example \eqref{exam3}, obtaining information about the system is a straightforward process. However, in this specific case, the existence of a time-varying nonlinear constraint creates difficulties in obtaining information about \eqref{smex10}, we are unable to utilize Proposition \ref{proposition2} to select $\gamma$. Nevertheless, the simplest and most direct option for $\gamma$ in this instance is $\gamma=0$ (if this choice proves ineffective, a larger value of $\gamma$ should be considered). As a result, the auxiliary ODE system corresponding to \eqref{smex10} can be expressed as follows:
$\left\{\begin{smallmatrix}
\dot{\xi}_{\gamma1}=-4\xi_{\gamma1}+(0.5\sin z) \nu_{\gamma},\\ 
\dot{\xi}_{\gamma2}=\frac{4}{3+\sin t}\xi_{\gamma1}-\frac{3+\cos t}{3+\sin t}\xi_{\gamma2},\\ 
\dot{\nu}_{\gamma}=-\frac{4}{3+0.5\cos z}\nu_{\gamma}.
\end{smallmatrix}\right.$
It follows that with $M=I$ we have
$J_M(t)=\sbm{
-4 & 0 &0.5\sin z \\ 
\frac{4}{3+\sin t} &-\frac{3+\cos t}{3+\sin t}& 0\\ 
0 & 0 &-\frac{4}{3+0.5\cos z}
}$.
Applying the 1-norm, we find that $\mu_1\Big(J_M(t)\Big)<-0.5$ holds, which implies that the DAE system \eqref{smex1} is contracting according to Theorem \ref{theorem1}. Hence, we can conclude that the time-varying DAE system is contracting. 
Figure \ref{fig:0} illustrates the trajectories of the system with two distinct initial conditions. 

\begingroup
\begin{figure}[!htb]
 \centering
 \includegraphics[width=1in]{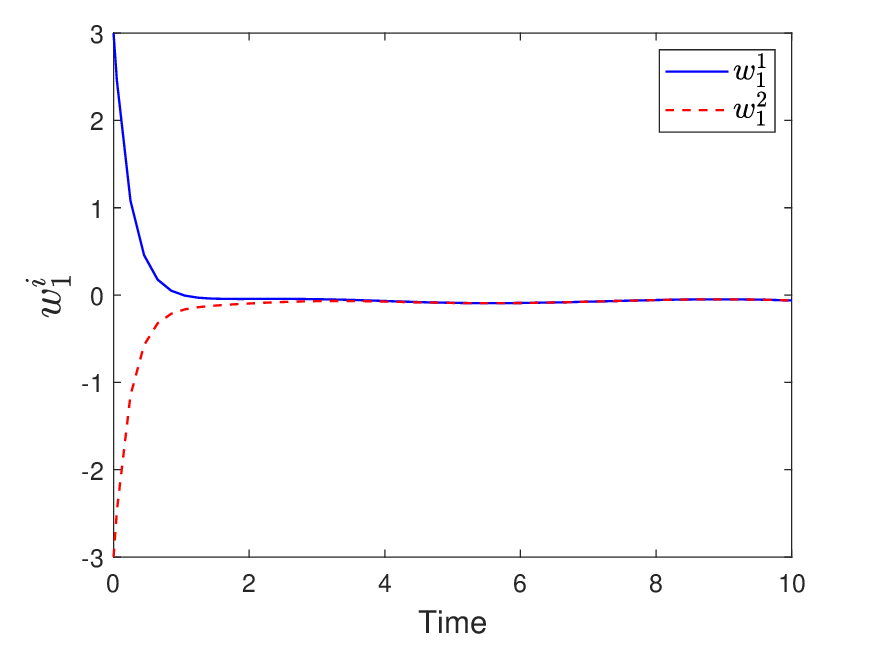}
 \includegraphics[width=1in]{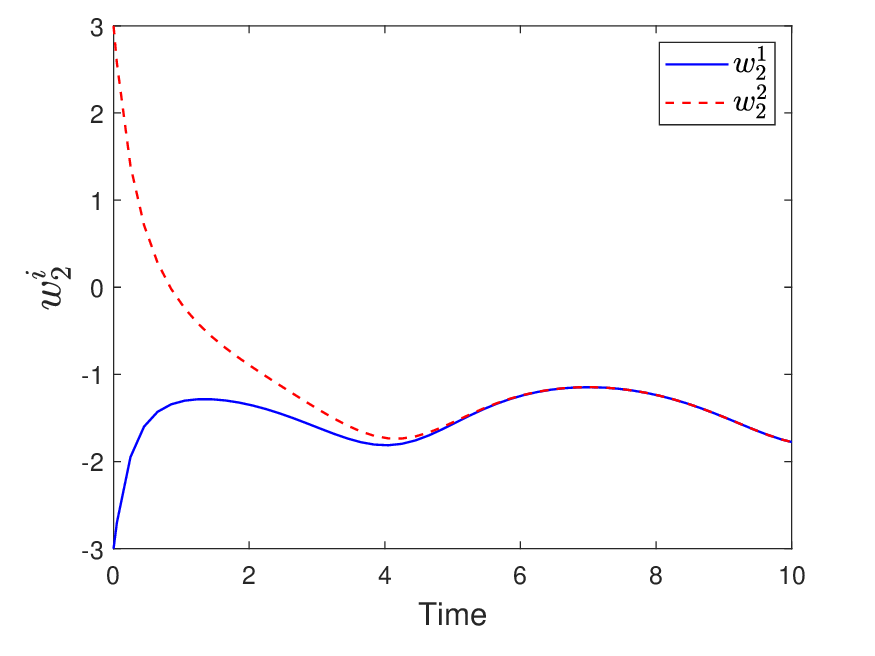}
  \includegraphics[width=1in]{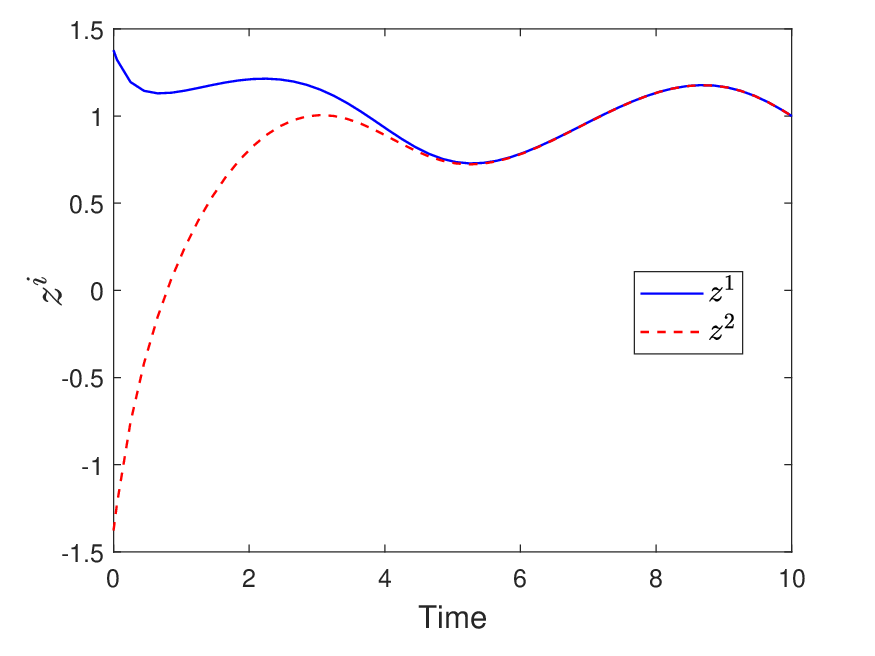}
  \caption{The plot of trajectories of time-varying DAE system in Example \ref{ex1}
initialized at $\sbm{3\\-3\\1.38}$ and  $\sbm{-3\\3\\-1.38}$.} \label{fig:0}
\end{figure}
\endgroup
\end{example}

\subsection{Stability of inverter-interfaced power source systems}
It is well known that for a contracting time-invariant system, all trajectories converge to an equilibrium exponentially. As an interesting particular case of our main results above, we can prove stability of the time-invariant DAE systems by using Proposition \ref{proposition1}. The basic idea involves guaranteeing the contravtivity of the DAE system while ensuring that its equilibrium point resides within its trajectory set. Consider a time-invariant DAE system given by
\begin{equation}\label{dae}
 \left\{\begin{matrix}
\dot{w}=f(w,z),\\ 
0=g(w,z),
\end{matrix}\right.   
\end{equation}
where $f(0,0)=0$, and $g(0,0)=0$. The variational system of \eqref{dae} is
\begin{equation}\label{variti}
 \left\{\begin{matrix}
\dot{\xi }=\frac{\partial f}{\partial w}(w(t),z(t))\cdot \xi +\frac{\partial f}{\partial z}(w(t),z(t))\cdot\nu,\\ 
0=\tfrac{\partial g}{\partial w}(w(t),z(t))\cdot\xi +\frac{\partial g}{\partial z}(w(t),z(t))\cdot\nu.
\end{matrix}\right.   
\end{equation}
To ensure safety in engineering systems, it is imperative to establish bounds on the states of \eqref{dae}. In the context of a power system, for example, the terminal complex voltage $V\angle \theta $ is restricted within the range $[0.95, 1.05]\times [-\frac{\pi}{2}, \frac{\pi}{2}]$ to ensure operational stability and prevent undesirable conditions. In this section, we investigate the time-invariant DAE system \eqref{dae} within the set $W\times Z$, i.e., $w(t)\in W\subseteq\mathbb{R}^{n}$ and $z(t)\in Z\subseteq\mathbb{R}^{m}$. Moreover, we assume that $\left[\frac{\partial g}{\partial z}\right]^{-1}\tfrac{\partial g}{\partial w}$ is bounded on $W\times Z$. In the subsequent corollary, we analyze the local exponential stability of \eqref{variti} through its reduced system $\dot{\xi}=\Big(\frac{\partial f}{\partial w}-\frac{\partial f}{\partial z}\left[\frac{\partial g}{\partial z}\right]^{-1}\tfrac{\partial g}{\partial w}\Big)\xi$.  


\begin{corollary}\label{corollary1}
The time-invariant DAE system \eqref{dae} is locally exponentially stable in $W\times Z$ if there exist an invertible matrix $M(w,z)$, such that $\mu_q\Big(J_M(w,z)\Big)\leq -\beta$ for some positive constant $\beta$, where $J(w,z)$ is given by 
\begin{equation}\label{j2}
\begin{aligned}
J_M(w,z)=\dot{M}M^{-1}+M \Big(\frac{\partial f}{\partial w}-\frac{\partial f}{\partial z}\left[\frac{\partial g}{\partial z}\right]^{-1}\tfrac{\partial g}{\partial w}\Big) M^{-1}.
\end{aligned}
\end{equation}
\end{corollary}
\begin{proof}
If condition \eqref{j2} is satisfied, then using similar arguments as in the proof of Theorem~\ref{theorem1} we can conclude that the reduced system is exponentially stable, i.e., $\left\|\xi\right\|\leq c\left\|\xi_0\right\|e^{-\alpha t}$ for some $\alpha,c>0$. By using $\nu = \left[\frac{\partial g}{\partial z}\right]^{-1}\tfrac{\partial g}{\partial w} \xi$ and the boundedness of $\left[\frac{\partial g}{\partial z}\right]^{-1}\tfrac{\partial g}{\partial w}$ on $W\times Z$, we can conclude $\left\|\sbm{\xi\\\nu}\right\|\leq c'\left\|\xi_0\right\|e^{-\alpha t}\leq c'\left\|\sbm{\xi_0\\\nu_0}\right\|e^{-\alpha t}$ for some $c'>0$. Consequently, the variational system \eqref{variti} is exponentially stable. Applying Proposition \ref{proposition1}, we can conclude that the DAE system \eqref{dae} is contracting. As $\sbm{w(t)\\z(t)} = \sbm{0\\0}$ represents one of admissible
trajectories of \eqref{dae} and it is contracting, it follows that all the trajectories will exponentially converge to $\sbm{w(t)\\z(t)} = \sbm{0\\0}$.
\end{proof}
 \begin{example}\label{ex2} Consider an inverter-interfaced power source
connected to the infinite bus via a transmission line \cite{Yang2021}. 
The system dynamics is given by the following
\begin{equation}\label{smex2}
\begin{aligned}
\left\{\begin{smallmatrix}
\dot{P}=\frac{1}{\tau_1}\Big(-P+P^{\text{ref}}-d_1(\theta-\theta^{\text{ref}})\Big),\\ 
\dot{Q}=\frac{1}{\tau_2}\Big(-Q+Q^{\text{ref}}-d_2(V-V^{\text{ref}})\Big)+u.
\end{smallmatrix}\right.
\end{aligned}
\end{equation}
where $P$, $Q$, and $u$ are the terminal output active, reactive power, and control input, respectively. The variable $V\angle \theta$ is the terminal complex voltage, $1\angle 0$ is the desired complex voltage. The symbols 
$P^{\text{ref}}$, $Q^{\text{ref}}$ , $\theta^{\text{ref}}$, and $V^{\text{ref}}$ are pre-specified constant reference values. The constants $\tau_1>0$ and $\tau_2>0$ are time constants while $d_1>0$ and $d_2>0$ are droop coefficients. For this particular case, the algebraic equations are given by
$\left\{\begin{smallmatrix}
P-GV\cos\theta-BV\sin\theta=0,  \\ 
Q-GV\sin\theta+BV\cos\theta=0.
\end{smallmatrix}\right.$. 
In this example, 
a state feedback controller $u=k_1P+k_2Q$ will be designed based on Corollary \ref{corollary1}. 
For numerical purposes, the parameters are given by $\tau_1=\tau_2=\frac{1}{3}$, $d_1=d_2=\frac{1}{3}$, $P^{\text{ref}}=1$, $Q^{\text{ref}}=-1$, $\theta^{\text{ref}}=0$, $V^{\text{ref}}=1$, $G=B=1$. In this case, the whole system can be rewritten as  
$\left\{\begin{smallmatrix}
\dot{P}=-3P-\theta+3,\\ 
\dot{Q}=-3Q-V-2+k_1P+k_2Q,
\end{smallmatrix}\right.$,
$\left\{\begin{smallmatrix}
P-V\cos\theta-V\sin\theta=0,  \\ 
Q-V\sin\theta+V\cos\theta=0,
\end{smallmatrix}\right.$,
where $W(t)=\sbm{
P\\ 
Q
}$, $Z(t)=\sbm{
\theta \\ 
V
}$. Its corresponding variational system is
$\left\{\begin{smallmatrix}
\delta \dot{P}=-3\delta P-\delta \theta,\\ 
\delta \dot{Q}=-3\delta Q-\delta V+k_1\delta P+k_2\delta Q,\\ 
\delta P-(\sin\theta +\cos \theta)\delta V-V(\cos \theta-\sin\theta)\delta \theta =0,\\ 
\delta Q-(\sin\theta -\cos \theta)\delta V-V(\cos \theta+\sin\theta)\delta \theta =0.
\end{smallmatrix}\right.$
The reduced system 
is
$\left\{\begin{smallmatrix}
\delta \dot{P}=(-3+\frac{\sin \theta -\cos \theta}{2V})\delta P-\frac{\sin \theta +\cos \theta}{2V}\delta \theta,\\ 
\delta \dot{Q}=(k_1-\frac{\sin \theta +\cos \theta}{2})\delta P+(-3+k_2-\frac{\sin \theta -\cos \theta}{2})\delta Q.
\end{smallmatrix}\right.$
In power system, we have $V\in[0.95,1.05]$, and $\theta\in[-\frac{\pi}{2}, \frac{\pi}{2}]$. Then, 
$\left\|\left[\frac{\partial g}{\partial z}\right]^{-1}\tfrac{\partial g}{\partial w}\right\|=\left\|\sbm{
\frac{\sin\theta -\cos\theta }{2V} & -\frac{\sin\theta+\cos\theta }{2V}\\ 
-\frac{\sin\theta+\cos\theta }{2} & -\frac{\sin\theta -\cos\theta }{2}
}\right\|$ is bounded. So we can analysis the stability of 
the DAE system by 
its reduced system. With $M=I$, we have 
$J_M(V,\theta )=\bbm{
-3+\frac{\sin \theta -\cos \theta}{2V}&-\frac{\sin \theta +\cos \theta}{2V}\\k_1-\frac{\sin \theta +\cos \theta}{2}&-3+k_2-\frac{\sin \theta -\cos \theta}{2}
}$.
We can determine the range of the elements in $J_M(V,\theta)$ within $[0.95,1.05]\times [-\frac{\pi}{2}, \frac{\pi}{2}]$. By setting $k_1=k_2=0.5$, and using the 1-norm, we find that $\mu_1\Big(J_M(V,\theta )\Big)<-1.1$. According to Corollary \ref{corollary1}, the DAE system 
is exponentially stable with respect to the equilibrium $[P^*,Q^*]^\top=[1,-1]^\top$. The trajectories of inverter-interfaced power source systems are depicted in Fig. \ref{fig:1}. With the state feedback controller, the terminal complex voltage $V\angle \theta$ converges to the desired complex voltage $1\angle 0$. 
\begingroup
\begin{figure}[!htb]
 \centering
 \includegraphics[width=1in]{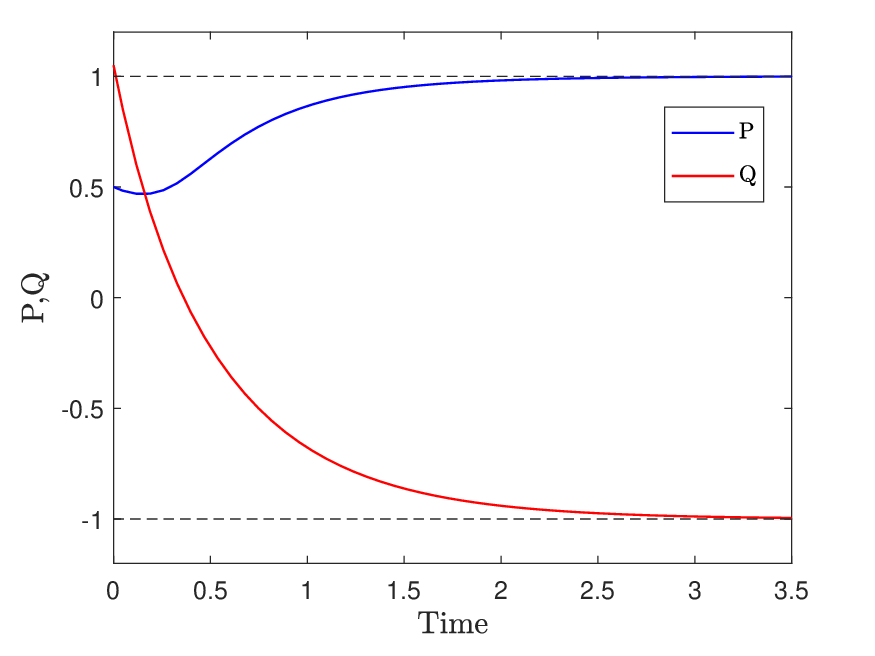}
 \includegraphics[width=1in]{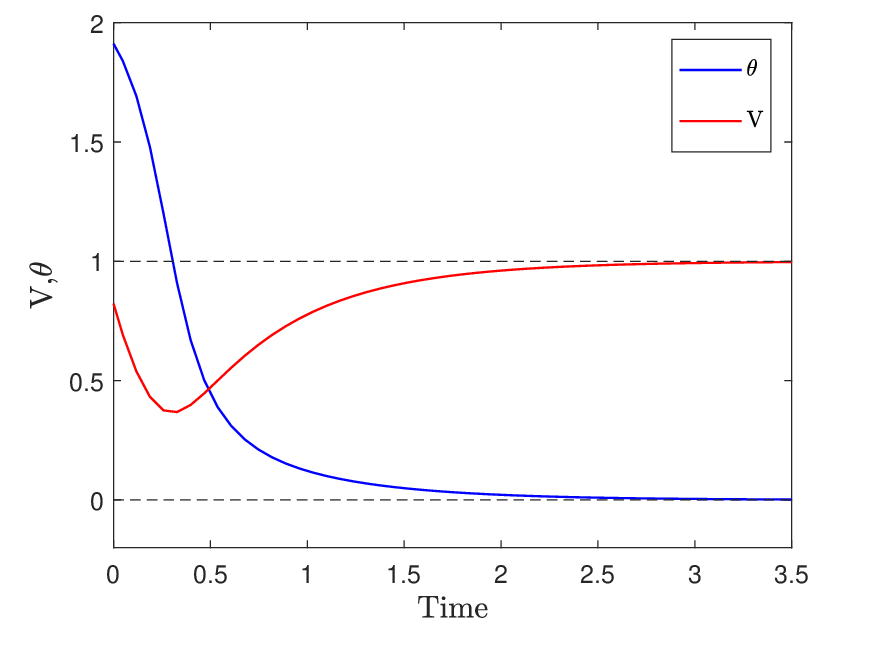}
  \caption{The plot of solutions of $P$, $Q$, $\theta$, and $V$ in Example \ref{ex2}
initialized at $\sbm{P_0\\Q_0}=\sbm{0.5\\1.05}$.} \label{fig:1}
\end{figure}
\endgroup
\end{example}

\subsection{Observer design of time-varying ODE systems}
By treating the output of the systems as an algebraic constraint, our approaches can be effectively employed in various classical control problems, including output feedback design, output regulation, and state observer design. In the following, we deploy our methodologies to develop an observer for a time-varying ODE system. Consider a time-varying ODE system
\begin{equation}\label{tvs}
 \left\{\begin{matrix}
\dot{w}=k(t,w),\\ 
z=h(t,w),
\end{matrix}\right.   
\end{equation}
where $w$ is the state and $z$ is the output. The design problem of interest is to determine an observer of the form 
\begin{equation}\label{tvo}
 \left\{\begin{matrix}
\dot{\hat{w}}=k(t,\hat{w})+l(t,\hat{z},z),\\ 
\hat{z}=h(t,\hat{w}),
\end{matrix}\right.   
\end{equation}
with $l(t,z,z)=0$, such that
\begin{equation}\label{tvo1}
\left\|\hat{w}-w\right\|\leq c\left\|\hat{w}_0-w_0\right\|e^{-\alpha t},   
\end{equation}
holds for all $t\geq 0$ 
with some positive constants $c$ and $\alpha$. In literature, a commonly used method for designing 
$l(t,\hat{z},z)$ is $l(t,\hat{z},z)=\kappa(t)(\hat{z}-z)$. As a result, the observer \eqref{tvo} can be simplified to the well-known Luenberger observer \cite{Rotella2013}.

The observer \eqref{tvo} for a given output signal $z(t)$  can be regarded as a time-varying DAE system in the form of \eqref{daet} with $f(t,\hat{w},\hat{z})=k(t,\hat{w})+l(t,\hat{z},z(t) )$, and $g(t,\hat{w},\hat{z})=h(t,\hat{w})-\hat{z}$. It is evident that $w$ and $z$ in \eqref{tvs} represent a solution of \eqref{tvo}. If \eqref{tvo} is contracting, both $w$ and $\hat{w}$ will satisfy \eqref{tvo1}. By applying Theorem \ref{theorem1} to \eqref{tvo}, we obtain the subsequent corollary.

\begin{corollary}\label{corollary0}
Given a system described by \eqref{tvs}, the system \eqref{tvo} is an observer of \eqref{tvs} if there exist an invertible matrix $M(t)$ and $\gamma>0$, such that $\mu_q\Big(J_M(t)\Big)\leq -\beta$ for some positive constant $\beta$, where $J_M(t)$ is in the form of \eqref{j1} with $A(t)=\frac{\partial k}{\partial \hat{w}}$, $B(t)=\frac{\partial l}{\partial \hat{z}}$, $F(t)=-I$, $C(t)=\gamma\frac{\partial h}{\partial \hat{w}}+\frac{\dd}{\dd t}(\frac{\partial h}{\partial \hat{w}})+\frac{\partial h}{\partial \hat{w}}\frac{\partial k}{\partial \hat{w}}$, and $D(t)=-\gamma I+\frac{\partial h}{\partial \hat{w}}\frac{\partial l}{\partial \hat{z}}$.
\end{corollary}\vspace{0.2cm}
\begin{example}\label{ex10}
Consider an unstable time-varying ODE system as presented in \cite[Ex.~4.22]{Khalil2002} 
\begin{equation}\label{oex1}
\begin{aligned}
\sbm{
\dot{w}_1\\ 
\dot{w}_2
}=\sbm{
-1+1.5\cos^2t & 1-1.5\sin t\cos t\\ 
-1-1.5\sin t\cos t & -1+1.5\sin^2t
}\sbm{
w_1\\ 
w_2
},
\end{aligned}
\end{equation}
with the output $z=w_1$. For simplifying the design process, we assume $l(t,\hat{z}-z)$ in \eqref{tvo} takes the form of $\sbm{
k_1(t)(\hat{z}-z)\\ 
k_2(t)(\hat{z}-z)
}$.
The Luenberger observer is given by
$\sbm{
\dot{\hat{w}}_1\\ 
\dot{\hat{w}}_2
}=\bbm{
-1+1.5\cos^2t & 1-1.5\sin t\cos t\\ 
-1-1.5\sin t\cos t & -1+1.5\sin^2t
}\sbm{
\hat{w}_1\\ 
\hat{w}_2
}+\sbm{
k_1(t)(\hat{z}-z)\\ 
k_2(t)(\hat{z}-z)
}$,
where the output $\hat{z}=\hat{w}_1$. By selecting $\gamma=1$, the auxiliary ODE system \eqref{auxode} of 
the observer is given by
$\left\{\begin{smallmatrix}
\dot{\xi }_{\gamma1}=(-1+1.5\cos^2t)\xi_{\gamma1}+(1-1.5\sin t\cos t)\xi_{\gamma2}+k_1(t)\nu_{\gamma}\\ 
\dot{\xi }_{\gamma2}=(-1-1.5\sin t\cos t)\xi_{\gamma1}+(-1+1.5\sin^2t)\xi_{\gamma2}+k_2(t)\nu_{\gamma}
\\ 
\dot{\nu}_{\gamma}=1.5\xi_{\gamma1}\cos^2t+(1-1.5\sin t\cos t)\xi_{\gamma2}+(k_1(t)-1)\nu_{\gamma}
\end{smallmatrix}\right.$.
By choosing $k_1(t)=-1.5\cos^2t$, $k_2(t)=-1+1.5\sin t\cos t$ it can be shown that inequality \eqref{ric} holds with $M=\sbm{
e^{-t}&0&0\\
0&e^{-t}&0\\
0&0&e^{-t}
}$
and by using the $2$-norm, we therefore have
$J_M(t)=\sbm{
-4+3\cos^2t& -3\sin t\cos t &0 \\ 
-3\sin t\cos t & -4+3\sin^2t & 0\\ 
0 & 0 &-3\cos^2t -4
}e^{-2t}\leq -I$.
The plots of error of $\hat{w}_i-w_i$ are shown in Fig. \ref{fig:3}, where the tracking property of the observer is ensured by the proposed methodologies.
\begingroup
\begin{figure}[!htb]
 \centering
 \includegraphics[width=1in]{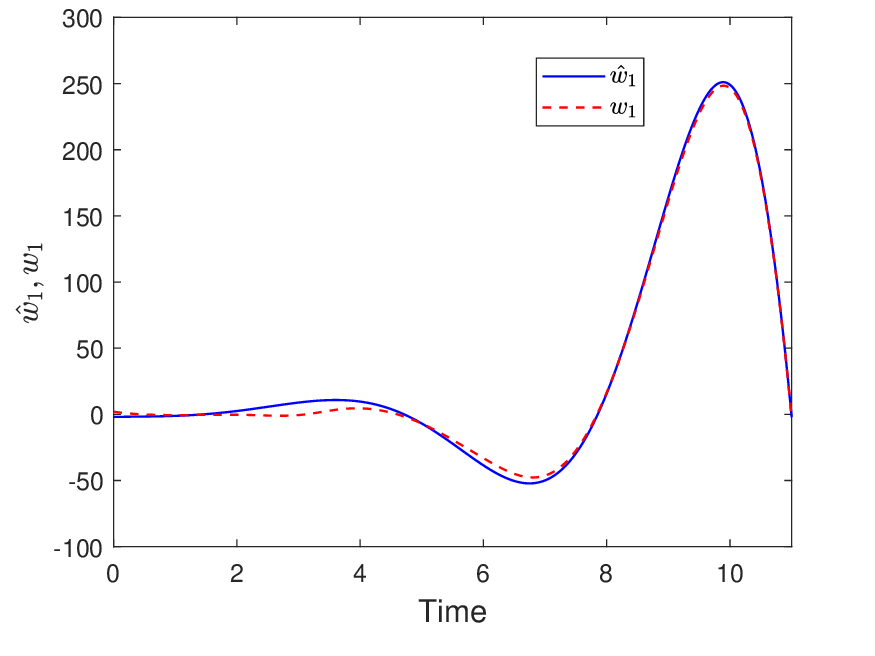}
 \includegraphics[width=1in]{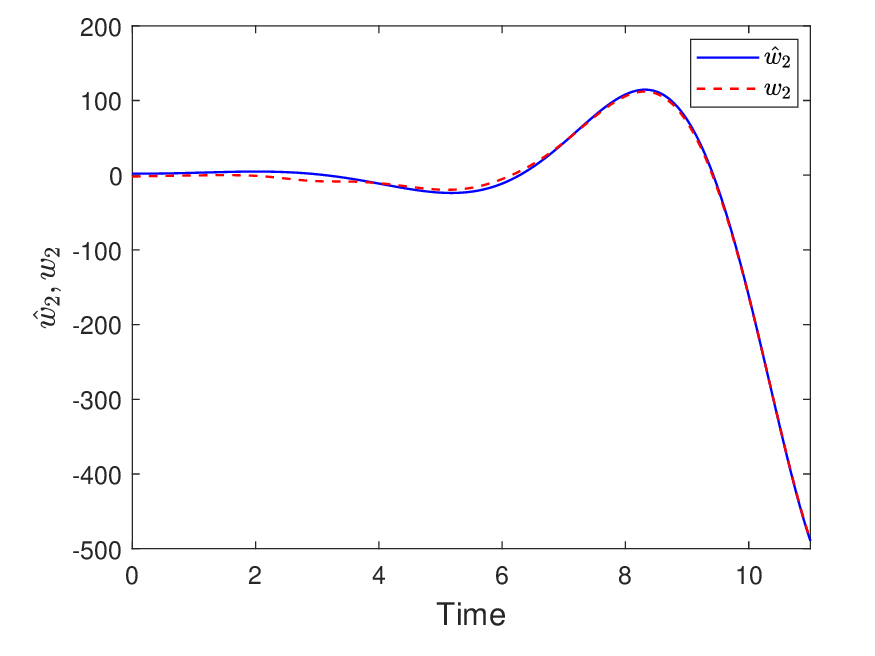}
 \includegraphics[width=1in]{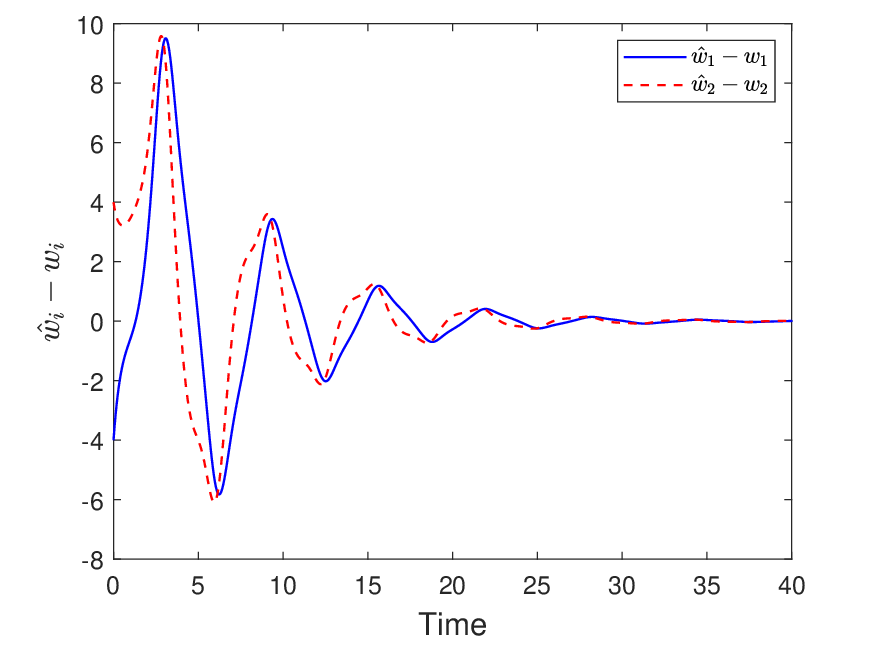}
  \caption{The plot of errors of $\hat{w}_i-w_i$ in Example \ref{ex10}
initialized at $\sbm{\hat{w}_1\\\hat{w}_2}=\sbm{2\\-2}$ and  $\sbm{w_1\\w_2}=\sbm{-2\\2}$.}\label{fig:3}
\end{figure}
\endgroup
\end{example}


\section{CONCLUSION}
In this paper, the contraction property of time-varying DAE systems have been studied by an ODE approach. It is established based on a necessary and sufficient condition that connects the contraction property of the original DAE systems and the UGES of its variational DAE system. Subsequently, the variational DAE system is lifted to an auxiliary ODE system defined on the whole state space, and the trajectories of these systems exhibit the same convergence property. The concept of matrix measure is introduced to study the UGES of the auxiliary ODE system. Moreover, the results obtained in this study can be applied to stabilize time-invariant DAE systems, and to observer design for time-varying ODE systems.

\end{document}